\documentclass[12pt,a4paper,final]{iopart}
\usepackage{amsthm,amsfonts,amssymb,amscd}
\usepackage{iopams}

\expandafter\let\csname equation*\endcsname\relax

\expandafter\let\csname endequation*\endcsname\relax

\usepackage{amsmath}
\usepackage{graphicx}
\usepackage{subcaption}

\usepackage[breaklinks=true,colorlinks=true,linkcolor=blue,urlcolor=blue,citecolor=blue]{hyperref}
\usepackage[utf8]{inputenc}	
\newcommand{\GO}[1]{ G^X_Y(z)_{#1}}
\newcommand{\GP}[2]{ G^{( #1 )}_Y(z)_{#2}}
\newcommand{\beq}{\begin{equation}}
\newcommand{\eeq}{\end{equation}}
\newcommand{\bea}{\begin{eqnarray}}
\newcommand{\eea}{\end{eqnarray}}
\newcommand{\half}{\frac{1}{2}}

\newcommand{\threehalves}{\frac{3}{2}}

\newtheorem{theorem}{Theorem}[section]
\newtheorem{proposition}[theorem]{Proposition}

\newtheorem{lemma}[theorem]{Lemma}

\begin{document}

\title[Boundary Conditions and the q-state Potts model on Random Planar Maps]{Boundary Conditions and the q-state Potts model on Random Planar Maps}

\author{Aravinth Kulanthaivelu$^{1}$, John F. Wheater$^{1}$}
\address{$^1$Rudolf Peierls Centre for Theoretical Physics, Clarendon Laboratory, Parks Road, University of Oxford, Oxford OX1 3PU, UK}
\ead{aravinth.kulanthaivelu@physics.ox.ac.uk, j.wheater@physics.ox.ac.uk}

\begin{abstract}
We extend a recent analysis of the $q$-states Potts model on an ensemble of random planar graphs 
with $p\leqslant q$ allowed, equally weighted, spins on a connected boundary.
In this paper we explore the $(q<4,p\leqslant q)$ parameter space of finite-sheeted resolvents and derive
the associated critical exponents. By definition a value of $q$ is allowed if there is a $p=1$ solution, and we
reproduce the long-known result that $q= 2(1+\cos{\frac{m}{n} \pi})$ with $m,n$ coprime. In addition we find that there are two distinct sequences of solutions, 
one of which contains $p=2$ and $p=q/2$ while the other does not. The boundary condition $p=3$ appears only for 
$q=3$ which also has a $p=3/2$ boundary condition; we conjecture that this new solution corresponds in the scaling limit to the 
`New' boundary condition, discovered on the flat lattice \cite{Affleck1998}. We also explore Kramers-Wannier duality for $q=3$ 
in this context and explicitly construct the known boundary conditions; we show that the mixed boundary condition
is dual to a boundary condition on dual graphs that corresponds to Affleck et al's identification of the 
New boundary condition on fixed lattices. On the other hand we find that the mixed 
boundary condition of the dual, and the corresponding New boundary condition of the original theory are not 
described by conventional resolvents.
\end{abstract}
\date{}
\vspace{2pc}
\maketitle

\section{Introduction/Overview}

There are two approaches to studying 2D Euclidean quantum gravity coupled to minimal matter. The continuum approach consists of Liouville CFT coupled to a minimal model CFT, while random graphs coupled to a spin system whose critical point is described by the CFT provide a discretised formulation. The matrix model provides a description of the latter in which the  fundamental degrees of freedom are the components of large $N$ Hermitian matrices taken from a random ensemble. The Feynman graph expansion of the matrix model yields the generating function for discretised random surfaces coupled to the spin degrees of freedom organised in a topological expansion.  
To make contact with the Liouville CFT, one must take an appropriate scaling limit, in which the discretised surfaces approach a continuum, and the matter system approaches a second order phase transition where the CFT describes the long range excitations of the spin system.

The one-matrix model formulates a discretisation of random surfaces only. These can be solved non-perturbatively through various methods \cite{Francesco1995} and are generally well understood. To add matter to the surfaces, more matrices are required.  A wide class of two-matrix models of this form are known to be solvable in the planar limit, where the surfaces have spherical topology, and there exist methods to compute subleading terms to all orders in the large $N$ expansion \cite{Eynard2003,Eynard2007}.
 Unfortunately, as first pointed out in \cite{Itzykson1980} and later in \cite{Kazakov2000},  multi-matrix models containing more than 2 degrees of freedom tend to be intractable. Those that can be solved use methods quite specific to the particular form of the model that do not generalise.

In this paper we study properties of the loop functions of the $q$-state Potts model coupled to 2D discretised quantum gravity, where the loop functions correspond to quantum gravity partition functions with a single boundary (i.e. a disk in the planar limit). We use the multi-matrix model formulation
\begin{equation}
\label{eq1}
d\mu(X_1,...,X_q) = \frac{1}{Z_{N,q}} \prod_{\langle i,j \rangle} e^{N \tr{X_i X_j}} \prod_{i=1}^q e^{-N\tr{V(X_i)}}dX_i.
\end{equation}
\noindent The $X_i$ are $N\times N$ Hermitian matrices, each representing one of the $q$ spin states, and $dX_i$ denotes integration over all the independent components. The potential $V(z)=U(z)-z^2/2$, where $U(z) = \sum_{k=2}^{n} t_m z^m/m$, controls the type of polygonulation of discrete 2-manifolds generated by the formal Feynman graph expansion in $\{t_2,t_3,\cdots\}$(we use $n=3$, corresponding to triangulations, as used in \cite{Carroll1996}).

At first this model seems intractable. The non-trivial cycle of couplings between matrices prevents simultaneous diagonalisation, so that some  techniques that are useful for single matrix models, such as orthogonal polynomials, are rendered inadequate \cite{Carroll1996a}. However, due to the particular structure of this partition function it is in fact solvable. Kazakov \cite{Kazakov1988} computed the solution for the case $q=1$ and $q=0$ limits and for $q=2$ it reduces to the well-known Ising model on a random lattice, for which the one-loop function has been computed  \cite{Kazakov1986,Boulatov1987}. In the Liouville CFT description, this corresponds to Liouville CFT coupled to the $c=1/2$ minimal model. Later, Daul and Zinn-Justin computed the one-loop function corresponding to the fully magnetised boundary condition on the spins and associated critical exponents for the $q=3$ case  \cite{DAUL1995,Zinn-Justin1998}. For all allowed values of $q$ Eynard and Bonnet \cite{Eynard1999} subsequently computed the one-loop function using the loop equation method.

In \cite{Atkin2015,Atkin2016} a formalism was developed to compute a larger set of loop amplitudes, including auxiliary boundary conditions on Potts spins as well as the free and partially magnetised conditions. This formulation was illustrated by application to the $q=1,2,3$ Potts models and the known results reproduced. In this paper we show how to use these methods to develop systematically the loop amplitudes for these boundary conditions at all allowed $q$-values, and compute the associated critical exponents.

This paper is organised as follows. In Section 2 we define the model and observables of the theory. In Section 3 we outline the method of determining the boundary generating (one loop) functions. In Section 4 we apply this construction to arbitrary $q$ and $p$, giving the allowed values of each, along with the degree of the discriminant in each case. Proof of the statements given in this section are largely left to the appendices. In Section 5 we outline some consequences of the previous propositions on the allowed boundary functions and derive the critical exponents of the general solution. In Section 6 we examine Kramers-Wannier duality on the random lattice, detailing the construction for both the Ising model and the 3-state Potts model, before concluding in Section 7.

\section{Defining the Model}

We use the model defined by \eref{eq1} to describe the $q$-state Potts model on a random planar lattice. 
The partition function $Z_{N,q}$ is given by the integral of \eqref{eq1} over all $q$ matrices. The free energy, given by 
\beq
F = \frac{1}{N^2} \ln{Z_{N,q}}
\eeq
\noindent is the sum over all closed triangulated 2D surfaces each weighted by the partition function of the Potts model on the corresponding dual graph. In the large $N$ limit it can be expanded in powers of $N$
\beq
F = \sum_{g=0}^{\infty} N^{-2g}F^{(g)}
\eeq
\noindent where $g$ denotes the genus of the diagrams that contribute to $F^{(g)}$. Thus the large $N$ expansion of the free energy becomes a topological expansion in the genus of the surfaces, and to leading order we only have planar diagrams \cite{Francesco1995,Ginsparg1993}.

Typically we are interested in computing the correlation functions of matrices. These are captured by the resolvents
\begin{equation}
\label{eq2}
    W_{(p)}(z) = \frac{1}{N}\langle \tr{\frac{1}{z-\sum_{k=1}^p X_k}}\rangle,
\end{equation}
\noindent where we use the $S_q$ symmetry of the model to identify the resolvent generated by $X_1$ with the resolvents generated by the other $X_i$. Therefore only sums of matrices matter. The resolvents are formal functions, defined through their asymptotic expansions in $z$. For example, the asymptotic expansion of $W_{(1)}(z)$ is given by
\beq
\label{eq3}
W_{(1)}(z) = \frac{1}{N} \langle \tr{\frac{1}{z-X_1}}\rangle = \sum_{k=0}^{\infty} \frac{T_{k}}{z^{k+1}}, \qquad T_k = \frac{1}{N} \langle \tr{X_1^k} \rangle.
\eeq
The $T_k$ moments, through the dual graph expansion, can be interpreted as the partition function of surfaces on which the boundary is composed of a length $k$ string of $X_1$ matrices. Therefore the $W_{(1)}(z)$ resolvent is a generating function for all discretised surfaces with a boundary of arbitrary length and containing only one type of matrix. For $W_{(p)}(z)$ with arbitrary orderings of $\{X_1,\cdots X_p\}$, the corresponding resolvent generates discretised surfaces with a boundary composed of the permitted matrices.

Interpreting the matrices as spins, we find that the resolvents generate the Potts model coupled to 2D surfaces with a single outer boundary (i.e. a disk for $g=0$), where the boundary admits a restricted subset of spins, given by the matrices used to define the relevant resolvent. For example, the single matrix resolvent of \eref{eq3} corresponds to the Potts model with a single spin on the boundary, and this would give the fixed spin boundary condition in the continuum CFT. In general we can map the resolvents onto the various conformal boundary conditions admitted by the Potts model CFT. We will return to some subtleties in this topic in Section 6. 

 For a given Hermitian matrix $X$ with eigenvalues $\{x_i\}_{i=1}^N$ we define the large $N$ eigenvalue density distribution
\begin{equation}
\label{eq4}
    \rho_X(x) = \lim_{N\rightarrow \infty}\frac{1}{N} \langle \sum_{i=1}^N \delta(x-x_i) \rangle,
\end{equation}
\noindent and the large $N$ resolvent is then given by the Stieltjes transform
\begin{equation}
\label{eq5}
    W_X (z) = \int_{\text{supp} \, \rho_X} dx \frac{\rho_X(x)}{z-x}, \quad z\notin \text{supp}\,  \rho_X.
\end{equation}

In the large $N$ limit, one can find the eigenvalue density distribution of the $X_i$, by computing the discontinuity of the genus zero resolvent along its compact (`physical') cut, and with this one can compute various correlation functions in the same limit. We can analogously define eigenvalue density distributions and large $N$ resolvents for any of the other resolvents given in \eref{eq2}. It is convenient to define the following functions, for a pair of Hermitian matrices $X$ and $Y$:
\bea
\label{eq7}
    G_Y^X(z) &=& \frac{1}{N} \frac{\partial}{\partial z} \ln{\langle \det_{1\leqslant k,l \leqslant N} e^{Nx_k y_l}\rangle_{y_N = z}}, \quad z\notin \text{supp}\,\rho_Y \\
%
\label{eq8}
    G_X^Y(z) &=& \frac{1}{N} \frac{\partial}{\partial z} \ln{\langle \det_{1\leqslant k,l \leqslant N} e^{Nx_k y_l}\rangle_{x_N = z}}, \quad z\notin \text{supp}\,\rho_X 
\eea
\noindent which satisfy the property $G_Y^X(G_X^Y(z)) = z + \mathcal{O}(1/N)$ \cite{Matytsin1994}. In the following we  determine the resolvents of the model in terms of these functions, to which  they are related  by an entire function
\begin{align}
\label{eq13}
    G_Y^X(z)_0 - G_Y^X(z)_{-1}=W_Y(z)_0 - W_Y(z)_{-1}, \\
    G_X^Y(z)_0 - G_X^Y(z)_{-1}=W_X(z)_0 - W_X(z)_{-1}, \nonumber
\end{align}
\noindent where the subscripts '0' refers to the physical sheet, and '-1' refers to the adjoining sheet, connected through a compact cut along the real axis.

\section{Boundary Generating Functions: Key Results}

In this section we recall the principal results of \cite{Atkin2016} and establish the procedure to determine the boundary generating functions. 
\begin{lemma}
Let $h>0$ and abbreviate the integral transformations
\bea
    \gamma_{\pm} (X) & = \int_{\mathbb{R}} dP_{\pm} f(P)e^{-\frac{N}{2}\tr{P_{\pm}^2}}e^{N\tr{P_{\pm}X/\sqrt{e^{\pm2h}-1}}} \label{eq9a}\\
    \gamma'_{\pm}(P)&=\int_{\Gamma} dX f(X) e^{N\tr{PX}\sqrt{1-e^{\mp 2h}}}\label{eq9b}
\eea
where the subscripts below the integrals indicate the integration cycle for the corresponding eigenvalues. Then, up to an overall factor, the partition function in \eref{eq1} can be written as
\begin{align}
\label{eq10a}
Z_{N,q} & = \int_{\mathbb{R}} dP_{+} e^{-\frac{N}{2} (1-e^{-2h})\tr{P_+^2}}(\gamma'_+ [e^{-N\tr{U}}](P_+))^q \\
& \label{eq10b} = \int_{\mathbb{R}} dX_0 \gamma_+ [(\gamma'_+ [e^{-N\tr{U}}])^p] (X_0) \gamma_- [(\gamma'_- [e^{-N\tr{U}}])^{q-p}] (X_0) \\
& \label{eq10c} = \int_{\mathbb{R}} dX_0 \bigg(\prod_{i=1}^q \int_\Gamma dX_i e^{-N\tr{U(X_i)}} \nonumber \\
& \quad \times \gamma_+[1]\bigg(X_0 + 2 \sinh{h}\sum_{i=1}^p X_i\bigg) \\
& \quad \times \gamma_-[1]\bigg(X_0 - 2\sinh{h}\sum_{i=p+1}^q X_i\bigg). \nonumber
\end{align}
\end{lemma}
With these integral transforms, the coupling terms between matrices in the action \eref{eq1} can be decomposed into simpler expressions with the introduction of new fiducial matrices. These new matrices, $P_+, P_-$ and $X_0$, couple to separate sums of the original matrices, as well as each other. We can then use the $S_{(q)}$ invariance of the model to isolate the relevant resolvent after application of the saddle-point method in the large $N$ limit. This leads to the following:
\begin{proposition}
Let the random matrix $P_+$ be defined as in Lemma 3.1, and set $Y=\sqrt{1-e^{-2h}}P_+$. Then for $N\rightarrow \infty$, the spectral density of the sum of $p$ matrices distributed according to \eref{eq1} is given by
\beq
\label{eq11}
\rho_{(p)}(z) = \frac{1}{2\pi i} [G_{(p)}^Y(z)_0 - G_{(p)}^Y(z)_{-1}], 
\eeq
where $G_{(p)}^Y(z)$ is the functional inverse of 
\beq
\label{eq12}
G_Y^{(p)}(z)_0 = \frac{p}{q}(z-W_Y(z)_{-1})+\frac{q-p}{q} W_Y (z)_0,
\eeq
\end{proposition}
\noindent where the superscript $(p)$ refers to the sum of matrices being studied, i.e. if we take $X_{(p)}=\sum_{k=1}^p X_k$, then $G^{X_{(p)}}=G^{(p)}$.

This argument rests on the fact that these $G$-functions have the same cut on the physical plane as the $W_Y(z)$ resolvents \eqref{eq13}. Using the result of this proposition, once we calculate the $G$-functions, we can determine the eigenvalue density distribution, from which one can calculate all of the observables of the theory.

To determine the $G$-functions, we use the saddle-point equation about the eigenvalues for a given $X$ matrix in \eref{eq1}:
\beq
\label{eq14}
U'(z) = W_X(z)_{-} + G_X^Y(z)_0.
\eeq
In analysing the $p=1$ case for integer $q$, we assume $W_X(z)_-$ has an asymptotic expansion such that $\lim_{z\rightarrow\infty}zW_X(z)=\mathcal{O}(1)$. This can be seen from the formal definition of the resolvent\eref{eq2}. Then, we can immediately infer from \eref{eq14} that $G_Y^X(z)$ has two sheets, connected by a semi-infinite square-root branch cut (taking $U(z)$ to be cubic), which we denote $C_\infty$. Furthermore Proposition 3.2 implies that $G_Y^X(z)$ has at least two sheets connected by a compact cut, which we denote $C_F$. If the resulting branch structure does not terminate there, we can analytically continue to the other sheets, and use Proposition 3.2 to write down expressions for $G_Y^X(z)$ on all of its other sheets. 

This derivation fixes the undetermined coefficients in the asymptotic expansion of the resolvent $W_X(z)$ by demanding that the $G$-functions have finitely many sheets. Starting from the physical sheet, `0', we analytically continue through $C_F$, to the next sheet, then $C_\infty$ to the sheet above, and repeating in this way, we eventually reach a sheet on which there is no longer a branch cut through which we can analytically continue the function to higher sheets.

\section{The $G$-functions for arbitrary $(p,q)$} 
\begin{proposition}
The only values of $q$ for which the $p=1$ generating function can be described by an algebraic equation are given by
\begin{equation}
\label{eq15}
    q= 2(1+\cos{\nu \pi}) \in (0,4),
\end{equation}
\noindent where $\nu = n/m$ and $n,m$ are coprime.
\end{proposition}
\begin{proof} This result was first observed by Eynard and Bonnet for the one-loop function in \cite{Eynard1999}. Here we establish it by computing the discontinuities across the branch cuts of the generating function and determining the values of $q$ for which they vanish. The analytic structure of $G_Y^X(z)$ is described in Fig \ref{fig:Gsheets}. Eliminating $W_Y(z)_{-1}$ between \eqref{eq12} and \eqref{eq13}, we have
\beq 
\label{eq16}
\GO{1}= z+(1-q)\GO{0}+(q-2)W_Y(z)_+ .
\eeq
Sheet $1$ is connected to sheet $0$ by $C_F$; circling $C_\infty$ gives
us
\beq 
\label{eq17}
\GO{2}= z+(1-q)\GO{-1}+(q-2)W_Y(z)_+ .
\eeq
Now we can eliminate $W_Y(z)_+$ between the last two equations to get
\beq 
\label{eq18}
\GO{2}= \GO{1}+(q-1)(\GO{0}-\GO{-1}).
\eeq
This demonstrates that we can generate all subsequent $\GO{K}$ with
$\GO{-1,0,1} $ as initial data.

\begin{figure}[h]
\centering
    \includegraphics[scale=0.85]{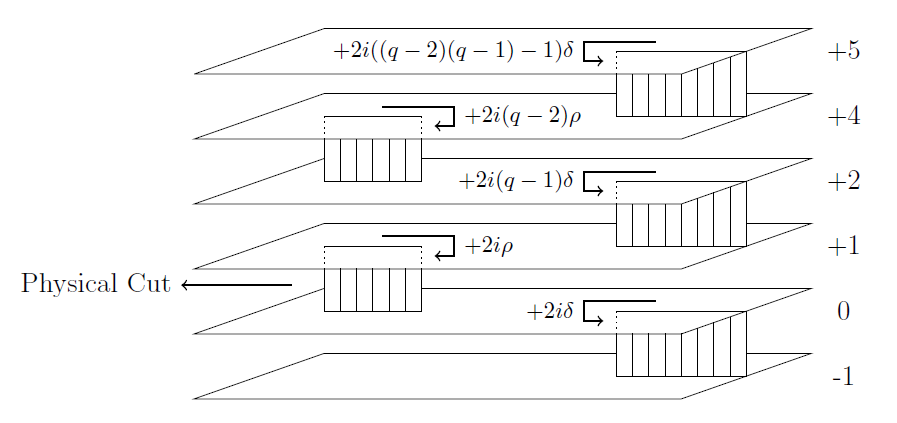}
    \caption{\label{fig:Gsheets} Analytic structure of $G_Y^X(z)$ for general $q$  in the matrix model with a cubic potential.}
\end{figure}

\noindent Circling successively $C_F, C_\infty, C_F, C_\infty,\ldots$ gives
\bea  
\label{eq19}
\GO{3} &=& \GO{0}+(q-1)(\GO{1}-\GO{-1}) \nonumber \\
\GO{4} &=& \GO{-1}+(q-1)(\GO{2}-\GO{0}) \nonumber \\
\GO{5} &=& \GO{-1}+(q-1)(\GO{3}-\GO{1}) \\
\vdots \nonumber \\
\GO{K} &=& \GO{K-6}+(q-1)(\GO{K-2}-\GO{K-4}). \nonumber
\eea
We note that this is a linear relationship amongst the $\GO{K}$s so we can read off the identical relationship for the coefficients of the discontinuous parts, $\rho_K$ and $\delta_K$ across the finite and infinite cuts respectively. We further note that the coefficient of any fixed power of $z$ also satisfies the same equation -- in particular  $\alpha_K$ , the coefficient of the $z$ itself will be important later in the computation of the discriminant.

From the above linear recursion relationship, we see that the $\rho_K$, $\delta_K$  and $\alpha_K$, defined above, all satisfy the difference equation
\beq 
\label{eq20}
y_K=(q-1)(y_{K-2}-y_{K-4})+ y_{K-6}, \quad K\geqslant 6
\eeq
with appropriate initial data generated by the expressions $\{y_0,y_2,y_4\}$ for even $K$, and $\{y_1,y_3,y_5\}$ for odd $K$. We can separate the solutions of this into those for which $K$ is an odd number and those for which $K$ is even. The solutions of this difference equation in the even case are of the form
\beq 
\label{eq21}
y_K= A x^{K/2},
\eeq
where $A$ is a constant and $x$ satisfies
\beq 
\label{eq22}
x^3-(q-1)(x^2-x)-1 =0,
\eeq
whose solutions are
\bea 
\label{eq23}
x&=&1, \nonumber \\
x&=&\frac{q-2\pm i \sqrt{4-(q-2)^2} } {2}\;=\; e^{\pm i\theta},
\eea
assuming $q\leqslant 4$. 

If the functions are algebraic, there must be a sheet for which the discontinuous parts $\rho_K$ or $\delta_K$ vanish. Using the boundary conditions for the difference equation, we have the following
\bea
\label{eq24}
    C_F&:& \begin{cases}  \rho_{2k} & = \frac{1}{\sqrt{q(q-4)}}x^{-(k+1)}(x^{2(k+1)}-1),\\
\rho_{2k+1} & = -\rho_{2k} 
\end{cases} \quad k\geqslant 1 \\
\label{eq25}
    C_{\infty}&:& \begin{cases}  \delta_{2k} & = \frac{1}{q-4}x^{-(k+1)}(x-1)(x^{2k+1}-1),\\
\delta_{2k-1} & = -\delta_{2k} 
\end{cases} \quad k\geqslant 1
\eea
%
The vanishing of $\rho_K$ or $\delta_K$ in \eqref{eq24} \eqref{eq25} fixes $x$ to be a root of unity whence, using \eqref{eq23}, we find that $q$ must satisfy \eref{eq15} and hence that $q<4$.
\end{proof}
 
\begin{proposition}
If $\theta=\nu\pi = \frac{n\pi}{m}$, with $n<m$ mutually prime, $n=1,3,\cdots$ is odd, and $m$ may be even or odd. Then the allowed values of $p$ are
\begin{equation}
\label{eq26}
    p=1+\frac{\sin{(M+1)\theta}}{\sin{M\theta}}, \quad M=1,\cdots, m-1
\end{equation}
\noindent if the sheets terminate on a compact cut, and
\begin{equation}
\label{eq27}
    p=1+\frac{\sin{(M+\frac{3}{2})\theta}}{\sin{(M+\frac{1}{2})\theta}}, \quad M=0,\cdots,m-1
\end{equation}
if the sheets terminate on an infinite cut. We denote this set of solutions Case 1.

If $\theta=\nu\pi=\frac{n\pi}{m}$, with $n<m$ mutually prime, $n=2,4,\cdots m-1$ is even, then the only allowed values of $p$ are those for which the sheets terminate on a compact cut:
\begin{equation}
\label{eq28}
    p=1+\frac{\sin{(M+1)\theta}}{\sin{M\theta}}, \quad M=1,\cdots, m-1
\end{equation}
We denote this set of solutions Case 2.
\end{proposition}
\begin{proof} The proof is straightforward and outlined in Appendix B.\end{proof}
\begin{proposition}
The degree of the discriminant of the $G$-functions, for all values of $(p,q)$, is given by
\begin{equation}
\label{eq29}
    \deg{\Delta(z)}=2m(2m-1)-m,
\end{equation}
in Case 1, and
\begin{equation}
\label{eq30}
    \deg{\Delta(z)} = m(m-1)-\frac{1}{2}(m-1),
\end{equation}
in Case 2, where $m$ is the number of sheets for the function.
\end{proposition}
\begin{proof} The proof is straightforward and outlined in Appendix C.
\end{proof}

Lastly, we outline what happens in the $q\geqslant 4$ cases. Now the sheet structure does not terminate and so a solution cannot be written in terms of algebraic functions. Nevertheless a solution can be obtained for the $q=4, p=1$ function by explicitly solving a double Riemann-Hilbert problem in terms of elliptic functions \cite{Zinn-Justin2000,Kazakov1999,Wenzel1967};  one then finds that the value of $p$ is unconstrained. In the case of $q> 4$ we find an ever increasing weighting for $C_F$ and $C_\infty$ which rules out any possibility of algebraic solutions. 

\section{Properties of the General Solution}
\subsection{Allowed Values of $p$ and the $q=3$ Boundary States}

By construction, the general $(q,p)$ solution allows a $p=1$ boundary function. However, our results show that other integer values of $p$ for a given allowed value of $q$ are not necessarily allowed even when $p<q$.
\begin{enumerate}
    \item $p=2$ is not allowed if $\theta=\nu\pi=\frac{n\pi}{m}$, with $n$ even, $n<m$ mutually prime (Case 2). From \eref{eq28} this would imply
\beq
\label{eq31}
1=\frac { \sin(M+1)\theta}{\sin M\theta}\Rightarrow
\theta =\frac{2\ell+1}{2M+1}\frac {\pi}{2},
\eeq
which is a contradiction.
\item $p=2$ is allowed if  $\theta=\nu\pi=\frac{n\pi}{m}$, with $n$ odd, $n<m$ mutually prime (Case 1). It appears in the series when $C_F$ vanishes \eref{eq26} if $m$ is odd in which case $M=\frac{m-1}{2}$, and it appears in the series for $C_\infty$ \eref{eq27} if $m$ is even in which case $M=\frac{m}{2}-1$.
\item $p=3$ is allowed only if $q=3$. We do not have a simple analytic proof of this, but have checked every case of $\theta=\frac{n\pi}{m}$ up to $m=170$. 
\end{enumerate}
There are two further values of $p$ which appear generically but are not in general integers:
\begin{enumerate}
    \item $p=q$ is always allowed. For Case 1, set $M=0$ in the $C_\infty$ series. For Case 2, set $M=\frac{m+1}{2}$ (in the $C_F$ series).
    \item $p=q/2$ occurs for Case 1:  if $m$ is even set $M=\frac{m}{2}$ in the $C_F$ series; if $m$ is odd set $M=\frac{m-1}{2}$  in the $C_\infty$ series. It is not allowed for Case 2 as is easily proved by contradiction.
\end{enumerate}

\subsection{Critical Exponents}

The critical point is the point in parameter space where, for the $G_Y^X(z)$ functions, the branch point corresponding to $C_\infty$ collides with a branch point of $C_F$. For the case of the $G_X^Y(z)$ functions, criticality is attained when a branch point collides with a critical point (i.e. where the derivative vanishes), and the square root singularity degenerates into a $(x-x_0)^{\frac{3}{2}}$ or higher order behaviour.

When we we have finite sheeted functions, the critical exponents can be computed directly from the discriminant,
\begin{equation}
\label{eq33}
\Delta(z):=\prod_{i< j} \big(G_Y^X(z)_i - G_Y^X(z)_j\big)^2.
\end{equation}
We use the following lemma, adapted from \cite{Dickey2003}.
\begin{lemma}
Let $z_0$ be a branch point of order $p$ over $z_0$, then the sum of orders of all branch points over some $z$ equals the multiplicity of $z_0$ as a root of the discriminant. Thus, the sum of orders of all branch points in the finite part of the Riemann surface equals the degree of the polynomial $\Delta(z)$.
\end{lemma}
\begin{proof}. According to the theorem on symmetric functions, the discriminant is a polynomial in $z$. Given a branch point of order $p$ over $z_0$, then $G_Y^X(z)$ on $p+1$ sheets will have the form 
\begin{equation}
\label{eq34}
G_Y^X(z)_i = a_0 + b(z-z_0)^{1/(p+1)}+\cdots,
\end{equation}
\noindent where the different branches of the $(p+1)$th root correspond to different $i$. There can also be other branch points over $z_0$. For all of them $a_0$ will be different. The contribution of the branch point to the discriminant is $(z-z_0)^{\frac{2}{p+1} \frac{(p+1)p}{2}}=(z-z_0)^p$. 
\end{proof}

By hypothesis, at criticality we have one central branch cut, about which the function behaves as $(z-z_0)^{r/s}$ for $r,s\in\mathbb{Z}$. The contribution of this branch point to the discriminant is then $(z-z_0)^{\frac{2r}{s}\sum_{i<j}} = (z-z_0)^{\frac{r}{s}n(n-1)}$, where $n$ is the number of sheets this branch point connects. Meanwhile we have contributions from the un-collided square-root branch cuts. For a given $n$ there will be $(n-2)/2$ of these for Case 1, and $(n-1)/2$ for Case 2. 
In Case 1, by the previous lemma, the degree of the discriminant will be given by
\begin{equation}
\label{eq35}
\deg \Delta(z) = \bigg(\frac{2m-2}{2}\bigg) + \frac{r}{s}2m(2m-1),
\end{equation}
while for Case 2, the degree of the discriminant will be given by
\begin{equation}
\label{eq36}
\deg \Delta(z) = \bigg(\frac{m-1}{2}\bigg) + \frac{r}{s}m(m-1).
\end{equation}
where $m$ is defined as in Proposition 4.2.

From \eref{eq29} and \eref{eq30} in the preceding section we have the degree of the discriminant, evaluated in the finite part of the Riemann surface. Comparing both, we find that the critical exponent in Case 1 is given by
\begin{equation}
\label{eq37}
\frac{r}{s} = \frac{2m-1}{2m},
\end{equation}
and in Case 2 it is given by
\begin{equation}
\label{eq38}
\frac{r}{s} = \frac{m-1}{m},
\end{equation}
which depends only on the number of sheets of the function. The string exponent is defined in terms of the resolvent as
\begin{equation}
    \label{eq39}
    W_{(p)}(z-z_c)\sim(z-z_c)^{1-\gamma_s}.
\end{equation}
Due to \eqref{eq13} the singular behaviour of $G_{(p)}^Y(z)$ about $C_F$ is the same as the behaviour of $W_{(p)}(z)$, and so we can relate our result for the former into a statement about the latter. Therefore we recover the known results $\gamma_s = -\frac{1}{2}, -\frac{1}{3}, -\frac{1}{5}$ for $q=1,2,3$, in agreement with the values found in \cite{Eynard1999}.

\section{The 3-State Potts Model and Kramers-Wannier Duality}

For the particular case of the 3-state Potts model, $n=1,\,m=3$,  there are four allowed values, $p=1,\, 3/2,\, 2, \,3$. The integer values of $p$ have a straightforward interpretation as boundary functions for which one, two or  three spin values are allowed respectively. However the $p=3/2$ case  has no easy physical interpretation that is local in terms of the microscopic theory.

The critical point of the model on a fixed lattice is described by a $(6,5)$ conformal field theory (CFT);  this is not part of the minimal series but has an extra $W_3$ symmetry and conserved current related to the $Z_3$ symmetry of the Potts model \cite{Fateev:1987vh}. The CFT has physical boundary states, Cardy states, corresponding to the boundary conditions that are invariant under boundary-preserving conformal transformations  \cite{Cardy:1984bb}. Affleck et al \cite{Affleck1998} showed that there are  eight such boundary states in this case. Seven states are accounted for by boundary conditions with a) fixed spins (three states, $Z_3$ triplet), b) mixed spins (three states,  $Z_3$ triplet), c) free spins (one state,  $Z_3$ singlet). They labelled the  eighth state, which is a  $Z_3$ singlet, as the `New' boundary condition.  The corresponding microscopic lattice picture involves negative Boltzman weights in the statistical mechanical model and lacks a simple intuitive physical interpretation (the quantum spin chain is more straightforward in this respect).

 It is tempting to conjecture that the $p=3/2$ state is, in the scaling limit, the random lattice analogue of the New boundary condition.  To confirm this, or otherwise, would require the extension of our methods to compute cylinder amplitudes. In the rest of this section we investigate the relationships between the boundary conditions in the original model, including the New boundary condition, and its dual. 
 
 \subsection{Ising Model}

Kramers-Wannier duality, originally described for 
the Ising model on flat lattices, provides a relationship between respectively the high-temperature and low-temperature expansions of the partition function on a lattice and its dual, see \cite{Kogut1979} for a review. On random lattices Kramers-Wannier duality was first studied in \cite{Carroll1996,Carroll1996a}. The situation is quite different from the flat lattice because the coordination numbers of the two lattices are radically different; one is finite (three in the case of the models considered in this paper), while the other is unconstrained.

\begin{figure}[h]
\centering
    \includegraphics[scale=0.5]{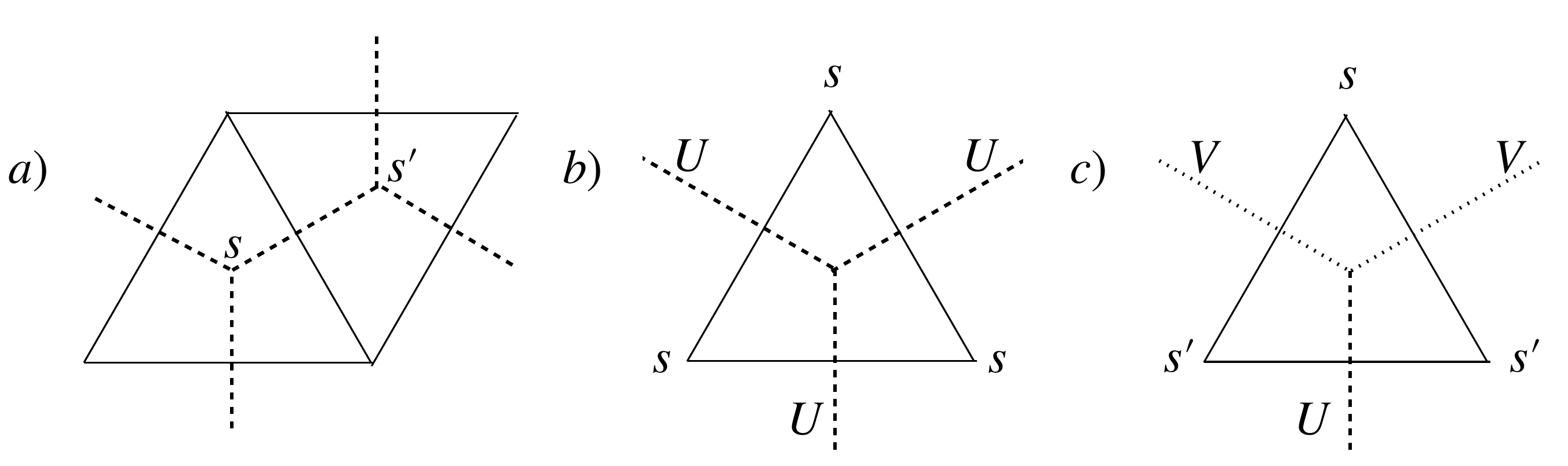}
    \caption{ \label{fig:duality} a) $\mathcal{M}_2$: spins $s,s'$ living on the faces of $T$; b) $\hat{\cal M}_2$: the vertex in $\hat T$ when all spins on the vertices of a triangle in $T$ take the same value; c) $\hat{\cal M}_2$: the vertex in $\hat T$ for $q=2$ when all spins on the vertices of a triangle in $T$ do not take the same value. }
\end{figure} 

In the graphs generated by (\ref{eq1}), the spins are located at the face of each polygon (triangle) in $T$, see Fig.~\ref{fig:duality}a. We associate weights 1 to an edge in the dual graph, $\hat T$, connecting identical spins, and $e^{-2\beta}$ to an edge connecting different spins. Starting with the Ising model, $q=2$, these weights are reproduced by the quadratic term in the matrix model action
\begin{equation}
   \mathcal{S}_2=\text{Tr}\bigg(\frac{1}{2}\frac{1}{1-e^{-4\beta}}(X_1^2 + X_2^2 - 2e^{-2\beta}X_1X_2) -\frac{g}{3}(X_1^3 +X_2^3 )\bigg),
\end{equation}
while the cubic term generates the vertices of $\hat T$; we refer to this model as $\mathcal{M}_2$.
By redefining the couplings and scaling the matrices, we recover the form given in \eqref{eq1}. In the dual theory, $\hat{\cal M}_2$, the Ising spins are located on the vertices of the triangles, rather than the faces. 
These combinatorics can be described by a new matrix model defined as follows:
\begin{enumerate}
    \item An edge connecting two vertices with the same spin carries weight 1 and is represented by a matrix $X$ in the dual lattice, Fig.\ref{fig:duality}b.
    \item An edge connecting two vertices with opposite spins carries weight $e^{-2\hat{\beta}}$ and is represented by a matrix $W$   in the dual lattice, Fig.\ref{fig:duality}c.
    \item There are two types of cubic vertex on the dual lattice corresponding to the cases when the three spins on the surrounding triangle are equal or not, see Figs.\ref{fig:duality}b \& 2c.
\end{enumerate}
This leads to the action
\begin{equation}
    \hat{\mathcal{S}}_2 = \text{Tr}\bigg(\frac{1}{2} X^2 + \frac{e^{2\hat{\beta}}}{2}W^2 - \frac{\hat{g}}{3}(X^3 + 3 X W^2)\bigg).
\end{equation}
Boundary configurations for the two models are different as shown in Fig \ref{fig:boundary}.  In $\hat{\cal M}_2$ a boundary condition is given by the spin values on the vertices of triangles lying on the boundary of $T$; whereas in ${\cal M}_2$ 
the spins are specified on the vertex of order one on the edge in $\hat T$ dual to the boundary edge of $T$.
\begin{figure}[h]
\centering
    \includegraphics[scale=0.5]{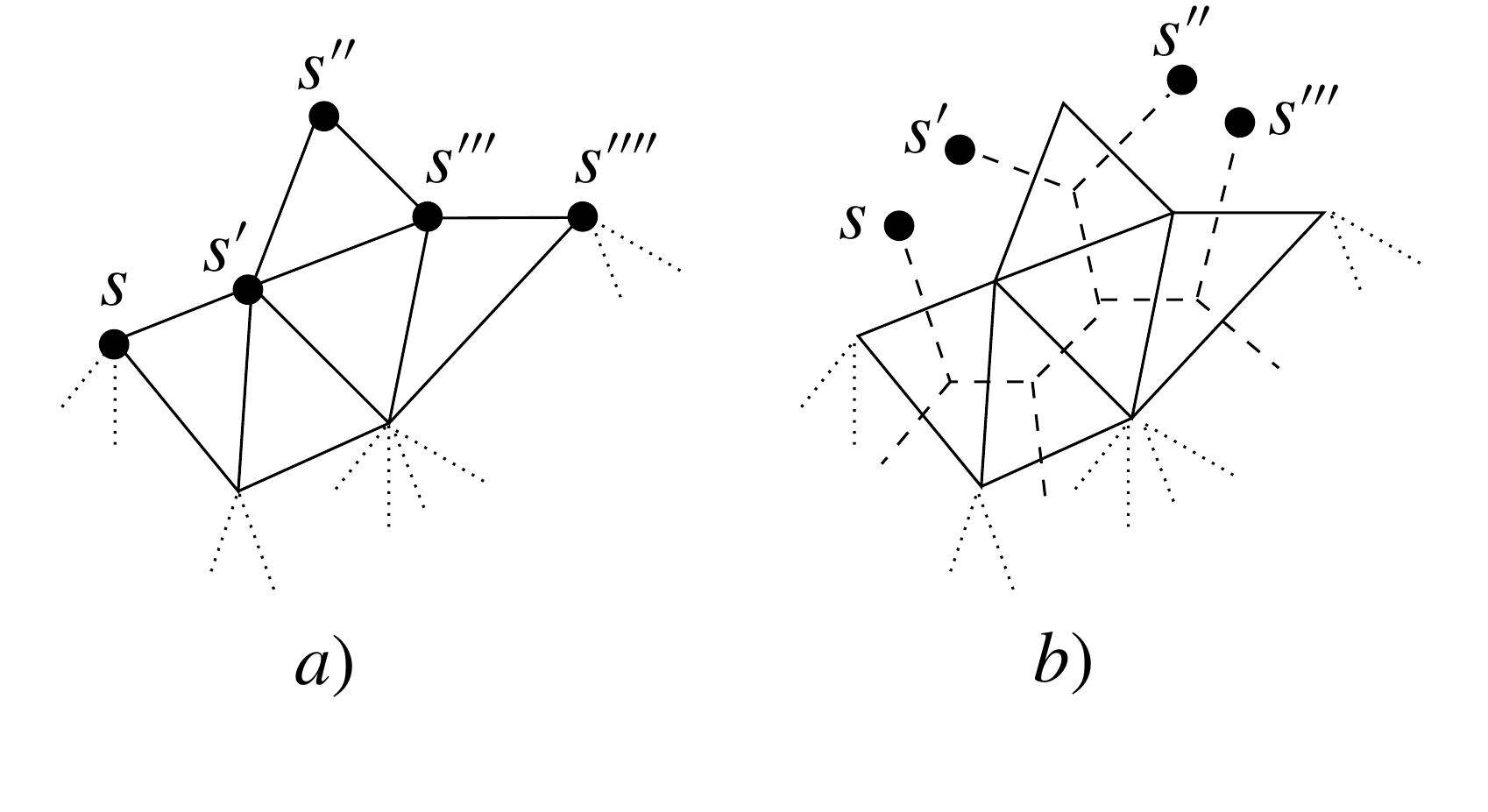}
    \caption{ \label{fig:boundary} Boundary spin configurations in a) $\hat{\cal M}_2$ and  b) ${\cal M}_2$. }
\end{figure} 

$\mathcal{S}_2$ and $\hat{\mathcal{S}}_2$ are related by the change of variables \cite{Carroll1996a} 
\begin{equation}
    X\rightarrow \frac{\lambda}{\sqrt{2}}(X_1+X_2),\quad W \rightarrow \frac{\lambda}{\sqrt{2}}(X_1-X_2), \quad \hat{g}\rightarrow \frac{\lambda^{-3}}{\sqrt{2}} g,
\end{equation}
where 
\begin{equation}
   \lambda=(1+e^{-2\beta}),\mathrm{~and~} \tanh\hat\beta =e^{-2\beta}.
\end{equation}
Note that the relationship between $\beta$ and $\hat\beta$ is the same as for a fixed lattice. The partition functions of ${\cal M}_2$ and $\hat{\cal M}_2$ are therefore equivalent but the relationship is more subtle for graphs with boundaries. The resolvent for the free boundary condition in $\mathcal{M}_2$, $W_{(2)}(z)$ can be calculated in the dual picture by computing the resolvent  $W_X(z)$, which describes the \emph{fixed} boundary condition in $\hat{\mathcal{M}}_2$, with the action $\hat S$. The two resolvents have the same scaling properties as each other and  as the fixed spin resolvent $W_{(1)}(z)$ \cite{Carroll1996a}. Similarly the resolvent  $W_W(z)$, describing a boundary with alternating spins in $\hat{\mathcal{M}}_2$, is equivalent in ${\mathcal{M}}_2$ to $W_{X_1-X_2}(z)$  which describes a system with an applied imaginary boundary magnetic field weighting the boundary spins by $e^{i\frac{\pi s}{2}}$. Note that this resolvent allows only boundaries of length 0 mod 2 and has a different scaling exponent from $W_{(1)}(z)$; in fact the resolvent $W_{Z_\alpha}(z)$ where
\begin{equation}
    Z_\alpha=(1-\alpha)X_1-(1+\alpha)X_2=W-\alpha X
\end{equation}
can be computed explicitly \cite{Carroll:1997tr,Atkin2011} and for $\alpha\ne 0$ has the same scaling exponent  as $W_{(1)}(z)$ so we conclude that the boundary condition described by $W$ is unstable in the infra-red. This is also the case for the fixed lattice Ising model, and consistent with the fact that the fixed and free boundary conditions appear as Cardy states in the $c=\half$ boundary CFT, but the alternating spin boundary condition does not \cite{Cardy:1989ir}.

\subsection{3-state Potts Model}

We can easily extend these ideas to the Potts model. The matrix model ${\mathcal{M}}_3$ for spins defined on the faces of the triangulations has action
\begin{equation}
     {\mathcal{S}}_3 = \text{Tr}\bigg(\frac{\mu(c)}{2}(\sum_{i=1}^3 M_i^2 - 2cM_1 M_2 -2c M_1 M_3 - 2c M_2 M_3) - \frac{g}{3}(\sum_{i=1}^3 M_i^3)\bigg),
\end{equation}
where 
\begin{equation}
    \mu(c)=\frac{(1-c)}{(1+c)(1-2c)}\;{\mathrm {and}}\; c=\frac{1}{e^\beta +1}.
\end{equation}
The Boltzmann weights are then
\begin{equation}\label{pottsboltzman}
    e^{\beta(\delta_{\sigma_k,\sigma_l}-1)} = \begin{cases}
      1 & \text{if $\sigma_k = \sigma_l$}\\
      \frac{c}{1-c} & \text{if $\sigma_k \neq \sigma_l$}
    \end{cases}      
\end{equation}
as on the flat lattice.

To construct the dual matrix model, map the spins on the vertices of the triangulation $T$ to phase factors $s\in\{1,\omega, \omega^2\}$, where $\omega=e^{i\frac{2\pi}{3}}$. To connect adjacent spins, we use the following matrices: $U$, which increases the phase by $\frac{2\pi}{3}$, $U^\dagger$ which increases the phase by $\frac{4\pi}{3}$ and $X$ which preserves the phase. The $X$ matrix has Boltzmann weight 1, since it connects identical spins, whereas $U$ and $U^\dagger$ have Boltzmann weight $e^{-\hat\beta}$, as they connect differing spins. Since the change in phase is dependent on the direction in which we go along the edge, the quadratic part of the action that measures the contribution from neighbouring vertices with different spins must be proportional to $UU^\dagger$. The cubic vertices can be derived by generalising Figs.\ref{fig:duality}b \& 2c to $q=3$.
These considerations lead to the following action for the dual model $\hat{\mathcal{M}}_3$
\begin{equation}
     \hat{\mathcal{S}}_3= \frac{1}{2}\text{Tr}(X^2 + 2e^{\hat\beta}U U^\dagger) - \hat{g}X(U U^\dagger + U^\dagger U) - \frac{\hat{g}}{3}(U^3 + U^{\dagger3} + X^3).
\end{equation}
As in the Ising model, ${\mathcal{M}}_3$ and $\hat{\mathcal{M}}_3$ are related through  a linear transformation of the matrix variables:
\begin{align}\label{PottsDualRelns}
    X = & \frac{\lambda}{\sqrt{3}}(M_1+M_2+M_3),   \quad U=  \frac{\lambda}{\sqrt{3}}(M_1+\omega M_2+\omega^2 M_3), \\ 
    U^\dagger =  &\frac{\lambda}{\sqrt{3}}(M_1+\omega^2 M_2+\omega M_3)
    ,  \quad \hat{g}=\frac{\lambda^{-3}}{\sqrt{3}}g, \quad  \lambda=\sqrt{\frac{1-c}{1+c}}, \nonumber
\end{align}
with
\begin{equation}
    e^{\hat\beta} =1+\frac{3}{e^\beta -1},
\end{equation}
the same relationship between temperature and dual temperature as  for the flat lattice.
Note that both ${\mathcal{S}}_3$ and $\hat{\mathcal{S}}_3$ inherit the permutation symmetry of the original Potts model through permutations  of $\{X_1,X_2,X_3\}$ and the exchange of $U$ and $U^\dagger$ respectively.

The loop functions we are interested in do not all have the same straightforward form in the dual model. For the fixed spin boundary in $\hat{\mathcal{M}}_3$, all spins on the boundary vertices are the same, see Fig.\ref{fig:pottsdualboundaries}a, which  corresponds to the resolvent $W_X(z)$; this is equivalent by \eref{PottsDualRelns} to the free boundary condition in ${\mathcal{M}}_3$. The free spin boundary condition in $\hat{\mathcal{M}}_3$ has graphs of the form Fig.\ref{fig:pottsdualboundaries}c which correspond to the resolvent  $W_{X+U+U^\dagger}$; this is equivalent by \eref{PottsDualRelns} to the free boundary condition in ${\mathcal{M}}_3$. 
\begin{figure}[h]
\centering
    \includegraphics[width=\textwidth]{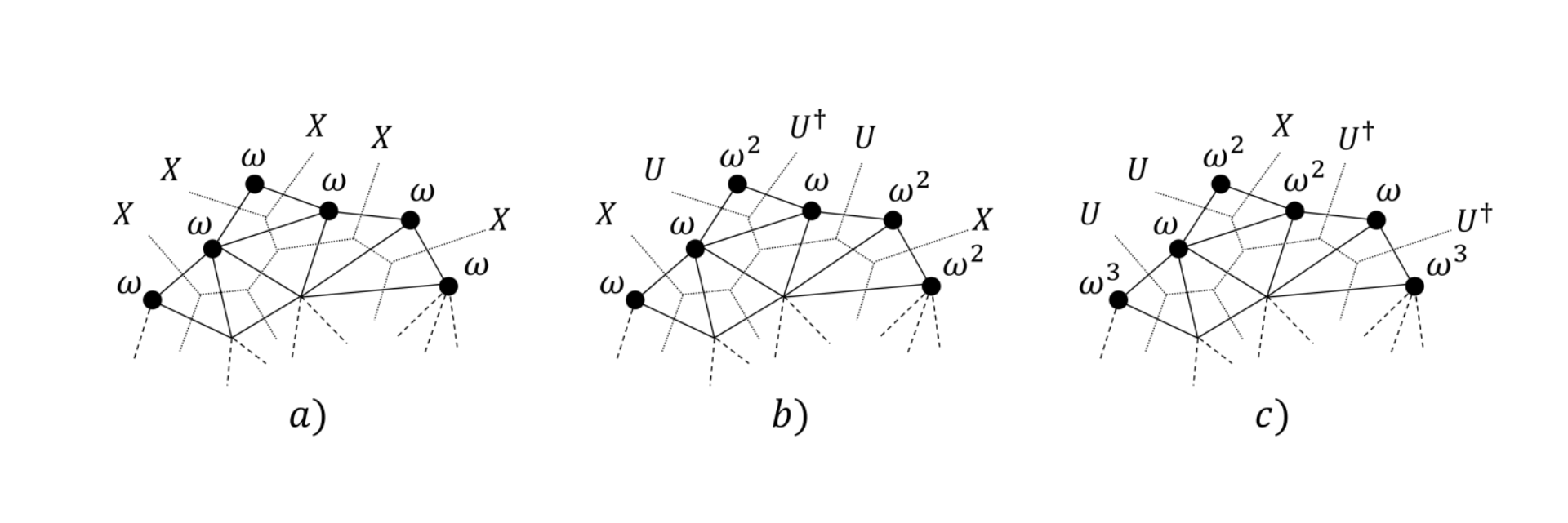} \caption{Examples of diagrams corresponding to various boundary conditions in $\hat{\mathcal{M}}_3$; a) Fixed spin, b) Mixed spin, c) Free spin}\label{fig:pottsdualboundaries}
\end{figure} 
The mixed boundary condition in $\hat{\mathcal{M}}_3$, Fig.\ref{fig:pottsdualboundaries}b, is expected to be dual to the New boundary condition in ${\mathcal{M}}_3$. 
The generating function for these  graphs is slightly harder to construct. To ensure that, for example, only $1$ and $\omega$ appear on the boundary $U$ can only be followed by $X$ or $U^\dagger$ and not another $U$, and there must be an equal number of $U$ and $U^\dagger$ for periodicity.
The loop function that accounts for all these features is
\begin{equation}
\label{newBoundary}
W_{\text{mixed}}(z)=    \frac{1}{N} \langle \text{Tr} \frac{1}{(z-(X+U\frac{1}{z-X}U^{\dagger})}\rangle.
\end{equation}
\noindent Using the matrix model technology available at the present moment we do not know how to compute this function. 

In \cite{Affleck1998} it was argued that the New boundary condition corresponds to deleting the Boltzman weights \eref{pottsboltzman} for edges of the graph belonging to the boundary and instead assigning the weights
\begin{equation}\label{pottsNew}
    w_{kl} = \begin{cases}
      1 & \text{if $\sigma_k = \sigma_l$,}\\
      -\frac{1}{2} & \text{if $\sigma_k \neq \sigma_l$.}
    \end{cases}      
\end{equation}
\begin{figure}[h]
    \centering
    \includegraphics[scale=0.42]{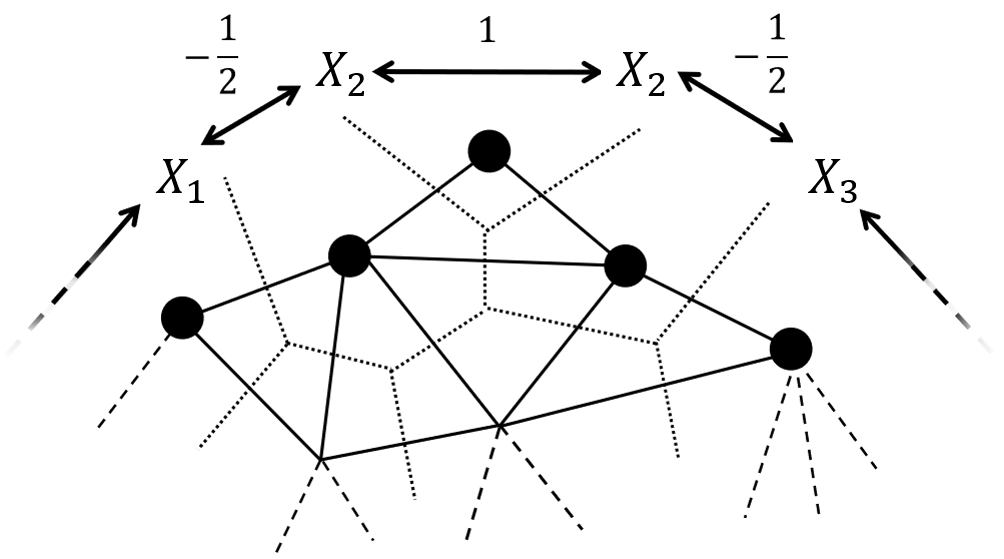}
    \caption{Example diagram in $\mathcal{M}_3$ corresponding to $\langle \text{Tr} \, \cdots X_1 X_2 X_2 X_3 \cdots \rangle$ with associated boundary interactions represented by solid lines between the spins $X_i$.}
    \label{fig:edge}
\end{figure}
For example, as shown in Fig.\ref{fig:edge}, neighbouring spins on the boundary of a given triangulation will contribute a factor of $-1/2$ when they differ and $1$ when they are the same. Therefore the amplitude will possess a weight of $(-1/2)^2$ due to the boundary interactions of this segment. This procedure is straightforward to implement in $\hat{\mathcal{M}}_3$ and gives the loop function
\begin{equation}
W_{\text{New}}(z)=    \frac{1}{N} \langle \text{Tr} \frac{1}{z-(X-\half U-\half U^{\dagger})}\rangle
\end{equation}
which, by \eref{PottsDualRelns}, is equivalent to the mixed boundary condition in ${\mathcal{M}}_3$. 

Similarly, it is possible to use the $S_3$ symmetry to show that this implementation holds when we map \eqref{PottsDualRelns} in $\hat{\mathcal{M}}_3$ to $\mathcal{M}_3$. To see this, we first note that the asymptotic expansion of \eqref{newBoundary} generates a restricted sum of length $n$ words in the free algebra generated by $\{X,U,U^\dag \}$. We then map these allowed words using \eqref{PottsDualRelns} to a weighted sum of words in the free algebra generated by $\{M_1,M_2,M_3\}$. Permutation symmetry with respect to exchange of $M_2,M_3$ means that only the real part of the expectation value of these words contribute. This leads to the following proposition:

\begin{proposition}
Under the mapping
\begin{equation}
    X = \sum_{\sigma \in \Delta} M_{\sigma}, \, U = \sum_{\sigma \in \Delta} \omega^{\sigma-1} M_{\sigma}, \, U^{\dag} = \sum_{\sigma \in \Delta} \omega^{1-\sigma} M_{\sigma},
    \label{propMap}
\end{equation}
the real part of the sum of allowed length $n$ words in the free algebra generated by $\{X,U,U^\dag\}$ maps to the following sum of length $n$ words in the free algebra generated by $\{M_1,M_2,M_3\}$,
\begin{equation}
\label{newWord}
   2^{n-1}\sum_{\sigma_1,\cdots,\sigma_n\in \Delta}\bigg(- \frac{1}{2}\bigg)^{\sum_{k=1}^{n}(1-\delta(\sigma_k,\sigma_{k+1}))}M_{\sigma_1}\cdots M_{\sigma_n}, \quad \sigma_{n+1} = \sigma_1,
\end{equation}
where $\Delta=\{1,2,3\}$. 
\end{proposition}
\begin{proof}
Let $W_n$ be the sum of allowed length $n$ words in $\{V,U,U^\dag\}$. Then for $n>0$,
\begin{equation}
    W_n = X W_{n-1} + \sum_{k=2}^n U X^{k-2} U^\dag W_{n-k},
\end{equation}
where we take $W_0=1$. Under \eqref{propMap} we have
\begin{equation}
    W_n = \sum_{\sigma_1,\cdots\sigma_n\in\Delta}C_n(\sigma_1,\cdots,\sigma_n)M_{\sigma_1}\cdots M_{\sigma_n}.
\end{equation}
The recurrence relation on the words can then be expressed as a recurrence relation on the coefficients $C_n(\sigma_1,\cdots \sigma_n)$,
\begin{equation}
\label{coRecurrence}
    C_n(\sigma_1,\cdots,\sigma_n)=C_{n-1}(\sigma_2,\cdots\sigma_n)+\sum_{k=2}^n \omega^{\sigma_1 - \sigma_k} C_{n-k} (\sigma_{k+1},\cdots,\sigma_n),
\end{equation}
where $C_1(\sigma)=C_0=1$.
We can now prove by induction that for $n>1$ this recurrence relation is equivalent to \footnote{We would like to thank an anonymous contributor to Mathematics StackExchange for drawing attention to this \href{https://math.stackexchange.com/questions/3379597/formula-for-the-sum-of-words-in-a-3-letter-algebra}{feature}.}
\begin{equation}
\label{ind}
    C_n(\sigma_1,\cdots,\sigma_n) = (1+\omega^{\sigma_1-\sigma_2})C_{n-1}(\sigma_2,\cdots,\sigma_n).
\end{equation}
For $n=2$ the recurrence relation gives
\begin{align}
    C_2(\sigma_1,\sigma_2) & = C_1(\sigma_2)+\omega^{\sigma_1-\sigma_2}C_0 \\
    & = (1+\omega^{\sigma_1-\sigma_2})C_1(\sigma_2).
\end{align}
Now assume the result holds up to some $n=l$. Then we must show that \eqref{coRecurrence} is equivalent to \eqref{ind} for $n=l+1$. By repeatedly applying this hypothesis we find
\begin{align}
    \omega^{\sigma_1-\sigma_2}C_{l}(\sigma_2,\cdots,\sigma_{l+1}) & = \omega^{\sigma_1-\sigma_2}C_{l-1}(\sigma_3,\cdots,\sigma_{l+1}) + \omega^{\sigma_1 - \sigma_3}C_{l-1}(\sigma_3,\cdots,\sigma_{l+1})\\
    & = \omega^{\sigma_1 - \sigma_2}C_{l-1}(\sigma_3,\cdots,\sigma_{l+1)}+\cdots  + \omega^{\sigma_1-\sigma_{k+1}}C_0 \\
    & = \sum_{k=2}^{l+1} \omega^{\sigma_1 - \sigma_{k}}C_{l+1-k}(\sigma_{k+1},\cdots,\sigma_{l+1}).
\end{align}
Hence the proposition holds for $n=l+1$. Thus, since it also holds for $n=2$, it is true for all $n$. 
Using this result we can easily solve the recursion,
\begin{equation}
    C_n(\sigma_1,\cdots,\sigma_n)=\prod_{k=1}^{n-1}(1+\omega^{\sigma_k-\sigma_{k+1}}).
\end{equation}
By taking the real part of this, we recover \eqref{newWord}.
\end{proof}
\noindent Therefore we find that the mixed boundary condition in $\hat{\mathcal{M}}_3$ is equivalent to a boundary condition in $\mathcal{M}_3$ with the weights for spins on the boundary given by \eqref{pottsNew}.

It is interesting to examine the boundary conditions generated by $W_U(z)$ and $W_{U^\dagger}(z)$; these are the ${\cal Z}_3$ analogues of the boundary condition generated by $W_W(z)$ in the Ising case. Evaluating the  resolvent series expansion and using the underlying $S_3$ symmetry under exchange of the matrices shows that only boundaries of length $0 \mod 3$ appear, analogous to the appearance of boundaries of length $0 \mod 2$ in $W_W(z)$. Crucially, as we observed above, this boundary condition has a different scaling dimension and flows to the free boundary condition under $Z_2$ symmetric boundary perturbations \cite{Atkin2011}. We expect the same behaviour to hold for the $q=3$ case; as soon as a perturbation $\epsilon X$ is added the restriction on boundary length disappears and the boundary condition will flow to the free boundary condition in the infra-red.

\section{Discussion}

We have classified all algebraic solutions for the disk partition function of the $q$-state Potts model coupled to planar random graphs/2D gravity with $p$ allowed spins on the boundary. We found that there is a discrete series of allowed values of $q$ that admit a $p=1$ boundary condition, in agreement with the same calculation using other methods \cite{Bernardi2011}, and the flat lattice result \cite{Itzykson1989}. For each of these we then determined the discrete series of allowed values of $p$ and we computed the discriminant of the general solution, using it to analyse the critical behaviour of the theory, and finding agreement with \cite{Eynard1996}. In the case of $q=3$, the $p=q/2$  boundary condition is in addition to the usual $p=1,2,3$ solutions. We conjectured that this extra boundary condition may play the role of the New boundary condition that appears in the fixed lattice spin system, but have not established such a connection which could only be true in the scaling limit as the $p=3/2$ boundary condition has no simple microscopic interpretation. 
We examined the Kramers-Wannier dual theory for $q=3$, extending the work of \cite{Carroll1996a} for $q=2$, and showed that the relationship between boundary conditions in the original and dual theories on a fixed lattice \cite{Fuchs1998} extends to the same statistical system on a random lattice. Within this framework we discovered that there is one boundary condition in the theory and the dual that is not a simple resolvent and we argued that this is the New boundary condition in the original theory, and the mixed boundary condition in the dual.

There are several ways we could extend this work. Firstly, it would be interesting to calculate resolvents in the dual theory explicitly. The matrix model is a mixed Hermitian-Complex matrix model with a particular structure that has not been studied in the literature. However, as the 3-state Potts matrix model is solvable, it is reasonable to expect that the dual theory is also a solvable matrix model and, in particular, that its mixed boundary condition loop function, $W_{\mathrm{mixed}}(z)$,  can be computed. Computing the dual loop functions would directly allow us to study their scaling behaviour and enable direct comparison with the loop functions of the original model. 

Secondly it would be interesting to calculate cylinder functions on the random lattice to  enable confirmation of our identification of boundary conditions, as in \cite{Atkin2011}. This would require extending our results to higher order in topology and pushing beyond the planar limit, as done for the $O(n)$ model in \cite{Borot2011},  using the algebraic curves found here as the starting point of the topological recursion procedure \cite{Eynard2003,Eynard2007}. 

Finally, we would like to understand the $q/2$ condition itself and determine what it corresponds to physically, if anything. If it is not the New boundary condition, then there is no apparent analogue for the fixed lattice. Discrepancies between fixed lattice and random lattice results have been observed in the past for the $O(n)$ model; in the strong coupling regime, corresponding to $n>2$ and central charge greater than one, new exotic critical points were discovered \cite{Durhuus1997,Eynard1996} which were not known to exist on the flat lattice. The Potts model seems to be related to the $O(n)$ model \cite{Eynard1992,Eynard1996}. We have found no new critical behaviour for $q<4$ in the $(q,p)$ parameter space of the $q$-state Potts model, but have not excluded an analogue of the new critical behaviour of the $O(n)$ model in the corresponding $q>4$ strong coupling regime.

\ack
We are grateful to A. Nahum for useful discussions. The work of AK is supported by the  STFC grant ST/N504233/1. 

\appendix

\section{$G$-function coefficients for arbitrary $p$ and $q$}

Going back to \eref{eq12} and \eref{eq13} we can eliminate $W_Y(z)$ between these two equations but keep $p$ general to find
\beq 
\label{eqA1}
\GP{p}{0}= \frac{p}{q}\big(z+   \GO{0}-\GO{1}\big)  + \bigg(1- \frac{2p}{q}\bigg) W_Y(z)_+ \, . 
\eeq
But $W_Y(z)_+$ also appears in \eref{eq16} so can be eliminated, and we obtain 
\beq 
\label{eqA2}
(q-2)\GP{p}{0}=(p-1)\big(z-\GO{1}\big) +(q-p-1) \GO{0} \, .
\eeq
Note that $q=2$ has to be dealt with separately but one can show that the results we will obtain are true even in that case. Circulating $C_F$ gives
\beq 
\label{eqA3}
(q-2)\GP{p}{1}=(p-1)\big(z-\GO{0}\big) +(q-p-1) \GO{1} \, .
\eeq
 Note that  formula \eref{eqA2} is trivially correct for $p=1$ and that
\begin{align}
\label{eqA4}(q-2)\GP{q-p}{0}&=(q-p-1)\big(z-\GO{1}\big) +(p-1) \GO{0}\\
	&= (q-2)z -(p-1)\big(z-\GO{0}\big) -(q-p-1)\GO{1} \nonumber\\
	&= (q-2)\big(z-\GP{p}{1}\big) ,\nonumber
\end{align}
which is the correct duality relationship between $p$ and $q-p$ functions, observed in \cite{Atkin2016}.

Going back to \eref{eqA2} circulating cuts in the sequence  $C_F, C_\infty, C_F, C_\infty,\ldots$ generates sheets with positive labels, so following on from \eref{eqA3} we have
\bea 
\label{eqA5}
(q-2)\GP{p}{2}&=&(p-1)\big(z-\GO{-1}\big) +(q-p-1) \GO{2} \nonumber\\
	(q-2)\GP{p}{3}&=&(p-1)\big(z-\GO{-1}\big) +(q-p-1) \GO{3} \nonumber\\
	(q-2)\GP{p}{4}&=&(p-1)\big(z-\GO{0}\big) +(q-p-1) \GO{4}\\
	\vdots\nonumber\\
	(q-2)\GP{p}{K}&=&(p-1)\big(z-\GO{K-4}\big) +(q-p-1) \GO{K} \, .\nonumber
\eea
And circulating cuts in the sequence  $C_\infty, C_F, C_\infty,C_F, \ldots$ generates sheets with negative labels, so
\bea 
\label{eqA6}
(q-2)\GP{p}{-1}&=&(p-1)\big(z-\GO{2}\big) +(q-p-1) \GO{-1}\nonumber\\
	(q-2)\GP{p}{-2}&=&(p-1)\big(z-\GO{3}\big) +(q-p-1) \GO{-1}\nonumber\\
	(q-2)\GP{p}{-3}&=&(p-1)\big(z-\GO{4}\big) +(q-p-1) \GO{0}\\
	\vdots\nonumber\\
	(q-2)\GP{p}{-K}&=&(p-1)\big(z-\GO{K+1}\big) +(q-p-1) \GO{K-3} \, .\nonumber
\eea
Note that, because the $\GO{K}$s and $y_K={\rm const}\cdot z $ all satisfy \eref{eq20}, so do the $\GP{p}{-K}$ and the only difference is the initial data. Therefore all $\GP{p}{-K}$ can be expressed as linear combinations of $\{z,\GO{-1,0,1} \}$; and the coefficients, $\rho^{(p)}_K$,  $\delta^{(p)}_K$ and $\alpha^{(p)}_K$  of the discontinuous parts and of $z$ respectively can be found.

Using these relationships it is straightforward to find that
\bea 
\label{eqA7}
\rho^{(p)}_{2M}&=&\frac{ (1-p)\sin M\theta+\sin(M+1)\theta}{\sin\theta}=-\rho^{(p)}_{2M+1},\\
\label{eqA8}	\delta^{(p)}_{2M+1}&=&-\frac{ (1-p)\sin (M+\half)\theta+\sin(M+\threehalves)\theta}{\sin\half\theta}=-\delta^{(p)}_{2M+2},\\
\label{eqA9}	 \rho^{(p)}_{-2M-1}&=&\frac{ (1-p)\sin (M+1)\theta+\sin M\theta}{\sin\theta}=- \rho^{(p)}_{-2M-2},\\
\label{eqA10}	\delta^{(p)}_{-2M}&=&-\frac{ (1-p)\sin (M+\half)\theta+\sin(M-\half)\theta}{\sin\half\theta}=-\delta^{(p)}_{-2M-1},
\eea
for $M=0,1,2,\ldots$.  For the coefficients of $z$ we get
\bea 
\label{eqA11}
\alpha^{(p)}_{2M +1}&=&\frac{-(p-2)\cos\half\theta+(p-1)\cos(M+\half)\theta   -\cos(M+\threehalves)\theta}{2\sin\half\theta\sin\theta},\nonumber\\
\alpha^{(p)}_{2M +1}&=&\alpha^{(p)}_{2M+2}, \qquad M=0,\pm 1,\pm 2\ldots\\
\alpha^{(p)}_{-2M }&=&\alpha^{(p)}_{-2M-1}.\nonumber
\eea

\section{Proof of Proposition 4.2}

For general $p$  the sheet structure can terminate: on an even positive label sheet in which case $\rho^{(p)}_{2M}=0$; on an odd positive label sheet  in which case $\delta^{(p)}_{2M+1}=0$; on an odd negative label sheet in which case $\rho^{(p)}_{-2M-1}=0$; on an even negative label sheet  in which case $\delta^{(p)}_{-2M}=0$. Clearly to be finite sheeted it must terminate on both a positive and a negative label sheet. The $\rho$'s and $\delta$'s have some general properties which help in classifying everything:
\begin{enumerate}
\item The absolute values of $\rho^{(p)}_{2M}$ etc have period $m$ in $M$. We need therefore only consider $0\le M < m$ in determining whether $\GP{p}{}$ is finite sheeted.
\item Note that if $\rho^{(p)}_{2M}=0$ then
\bea 
\label{eqB1}
\rho^{(p)}_{-2(m-M-1)-1}&=&\frac{ (1-p)\sin (m-M)\theta+\sin (m-M-1)\theta}{\sin\theta}, \nonumber\\
					&=&\frac{ (1-p)\sin (n\pi-M\theta)+\sin (n\pi -(M+1)\theta)}{\sin\theta},\\
					&=& (-1)^{n+1}\rho^{(p)}_{2M} =0 \, . \nonumber
\eea
Similarly if $\delta^{(p)}_{2M+1}=0$ then $\delta^{(p)}_{-2(m-M-1)}=0$. It follows that if the structure terminates at a positive label sheet it will definitely also terminate at a negative label sheet (though these results do not necessarily identify the lowest label sheet on which the structure actually terminates).
\item In Case 2 observe that 
\begin{align}
\label{eqB2}
\rho^{(p)}_{2(M+1)}&=(-1)^{n/2}\frac{ (1-p)\sin ((M+1)\frac{n\pi}{m}-\frac{n\pi}{2})+\sin ((M+2)\frac{n\pi}{m}-\frac{n\pi}{2})}{\sin\theta},\nonumber\\
&=(-1)^{n/2}\frac{ (1-p)\sin ((M+1-\frac{m}{2})\frac{n\pi}{m})+\sin ((M+2-\frac{m}{2})\frac{n\pi}{m})}{\sin\theta},\\
&=(-1)^{n/2}\frac{ (1-p)\sin ((M-\frac{m-1}{2} +\half)\frac{n\pi}{m})+\sin ((M-\frac{m-1}{2} +\threehalves)\frac{n\pi}{m})}{\sin\theta}.\nonumber
\end{align}
But $m$ is odd so $k=\frac{m-1}{2}$ is an integer and we find
\bea 
\label{eqB3}
\rho^{(p)}_{2(M+1)}=(-1)^{n/2}\delta^{(p)}_{2(M-k)+1},
\eea
and similarly, that
\bea 
\label{eqB4}
\rho^{(p)}_{-2M-1}=(-1)^{n/2}\delta^{(p)}_{-2(k-M)}.
\eea
It follows that termination on even or odd sheets generates the same set of $p$ values for which $\GP{p}{}$ is finite sheeted; whether it actually terminates on an even or an odd sheet is determined by which one has lower label.

\end{enumerate}

\subsection{\emph{Case 1}  $\theta=\nu\pi=\frac{n\pi}{m}$, $n < m$ are mutually prime,  $n=1,3,\ldots $ is odd, $m$ may be even or odd}

By properties (ii) it follows that all $p$ values for which $\GP{p}{}$ is finite sheeted are given by
$\rho^{(p)}_{2M}=0$, $M>0$ or by $\delta^{(p)}_{2M+1}=0$, $M\ge 0$ giving two sequences
\bea \label{eqB5}{~~~ S_1:~}p&=&1+\frac { \sin(M+1)\theta}{\sin M\theta},\qquad M=1,\ldots m-1\\
\label{eqB6}\text{or}{~ S_2:~}p&=&1+\frac { \sin(M+\threehalves)\theta}{\sin( M+\half)\theta},\qquad M=0,\ldots m-1
\eea
where $S_1$, for $M=1,\ldots m-1$, terminates on sheet $2M$, ie $2,4,\ldots 2(m-1)$ and \\
$S_2$, for $M=0,\ldots m-1$, terminates on sheet $2M+1$, ie $1,3,\ldots 2m-1$.

Note that $S_1$ and $S_2$ are distinct, so we can find the negative label on which they terminate  by inserting the values of $p$ from \eref{eqB5} and \eref{eqB6} respectively into the expressions $\rho^{(p)}_{-2M'-1}$ \eref{eqA9} and $\delta^{(p)}_{-2M'}$ \eref{eqA10} and solving for $M'$. This gives\smallskip\\
$S_1$ $M=1,\ldots m-1$ terminates on sheet $-2(m-M)+1$, ie $-2m+3,-2m+5,\ldots -3,-1$,\smallskip\\
$S_2$ $M=0,\ldots m-1$ terminates on sheet $-2(m-M-1)$, ie $-2m+2,-2m+4,\ldots -2,0$.

Now we can show that $S_1\ne S_2$ are distinct. If they were not then we would have for some $M,M'$
\bea 
\label{eqB7} \sin( M+\threehalves)\theta \;\sin M'\theta &=& \sin( M+\half)\theta\;\sin(M'+1)\theta,\\
\Rightarrow (2(M-M') +1)\frac{n\pi}{m}&=&2\ell\pi,
\eea
which is a contradiction because $n$ is odd.

\subsection{\emph{Case 2}  $\theta=\nu\pi=\frac{n\pi}{m}$, $n< m$ are mutually prime and $n=2,4, \ldots m-1$ is even}
By properties (ii) and (iii) it follows that all $p$ values for which $\GP{p}{}$ is finite sheeted are given by
$\rho^{(p)}_{2M}=0$, $M>0$ so in this case there is only one sequence given by
\beq 
\label{eqB8}
p=1+\frac { \sin(M+1)\theta}{\sin M\theta},\qquad M=1,\ldots m-1
\eeq
(note that the special case $p=1$ is $M=m-1$). 

By property (iii) we see for $M=\frac{m+1}{2} ,\ldots m-1$ that $\delta^{(p)}$ is zero for a lower labelled sheet than $\rho^{(p)}$, so for positive labels
\begin{itemize}
    \item $M=1\ldots\frac{m-1}{2}$  terminates~on~sheet $2M$,  ie $2,4,\ldots m-1$\smallskip\\
    \item $M=\frac{m+1}{2}\ldots m-1$  terminates~on~sheet $2M-m$  ie $1,3,\ldots m-2$
\end{itemize}

Using property (i) we  see that for negative labels
\begin{itemize}
    \item $M=1\ldots\frac{m-1}{2}$  terminates~on~sheet  $-m+2M+1$  ie $-m+3,\ldots 0$\smallskip\\
    \item $M=\frac{m+1}{2}\ldots m-1$  terminates~on~sheet $-2(m-M)+1$  ie $ -m+2,\ldots,-3,-1$
\end{itemize}

\section{Proof of Proposition 4.3}
We can examine the $\alpha_K$ \eref{eqA11}, the coefficient of $z$ in the recurrence relation for the $G$-functions, and determine the  degree of the discriminant.
\begin{enumerate}
\item First note that $\alpha_{-1} =\alpha_{0}=0$ and that in general $ \alpha_{2M +1}=\alpha_{2M+2}$ so that $\GO{2M+2}-\GO{2M+1} = {\rm const}\cdot \sqrt{z} + \text{lower order terms}$.
\item Now we prove that there are no other cases where $ \alpha_{K}=\alpha_{K'}$ in the finite sheet structure. For $ \alpha_{2M+1}=\alpha_{2M'+1}$ we require
\beq \cos(M+\threehalves)\theta=\cos(M'+\threehalves)\theta, \eeq
with $M,M'=-1,\dots m-2$ in case 1, or $M,M'=-1,\dots \half(m-3)$ in case 2. The first possibility is that $M=M'+m$ if $n$ is even or $M=M'+2m$ if $n$ is odd; but then $M$ falls outside the allowed range. The second possibility is that
\beq (M+\threehalves)\theta=2N\pi-(M'+\threehalves)\theta,\qquad N= 1,2,\ldots \eeq
so that 
\beq (M+M'+3)\frac{n}{m} = 2N.\eeq
 In Case 2 this is satisfied if  $M+M'+3=m$ which can only be true if $M=M'= \half(m-3)$ which proves the result. In Case 1 then $M+M'+3=2m$ which can only be true if one of $M,M'$ is $m-1$ which is out  of range and proves the result.
\end{enumerate}

\noindent Now we can compute the leading power of $z$ in the discriminant. In Case 1 
\begin{align} 
\Delta(z) &= \prod_{K<K'} \big( \GO{K}-\GO{K'} \big)^2,\\
		&= \prod_{M=-1}^{m-2} \big( \GO{2M+2}-\GO{2M+1} \big)^2 \prod_{K'-K>1} \big( \GO{K}-\GO{K'} \big)^2.
\end{align}
So for $2m$ sheets we have $m$ terms in the product giving a factor $z$  and  $2m(2m-1)/2-m$ terms giving a factor $z^2$. Therefore
\beq {\rm deg} \Delta(z) = 2m(2m-1)-m.\eeq
In Case 2
\begin{align} 
\Delta(z) &= \prod_{K<K'} \big( \GO{K}-\GO{K'} \big)^2,\\
		&= \prod_{M=-1}^{\half(m-5)} \big( \GO{2M+2}-\GO{2M+1} \big)^2 \prod_{K'-K>1}\big( \GO{K}-\GO{K'} \big)^2,
\end{align}
this time there are $m$ sheets  and we have $\half(m-1)$ terms in the product giving a factor $z$
and  $m(m-1)/2-\half(m-1)$ terms giving a factor $z^2$.
Therefore
\beq {\rm deg} \Delta(z)= m(m-1)-\half(m-1).\eeq

Next we can examine the $\alpha^{(p)}_K$ and show that the  degree of the discriminant is the same 
as it is for $p=1$ for all values of $p$ for which $\GP{p}{}$ is finite  sheeted.
\begin{enumerate}
\item First note that $\alpha^{(p)}_{-1} =\alpha^{(p)}_{0}=0$ and that $\alpha^{(p)}_{L}\ne 0$ if $L\ne 0,-1$, as is easily shown from (39), and that in general $ \alpha^{(p)}_{2M +1}=\alpha^{(p)}_{2M+2}$ so that $\GP{p}{2M+2}-\GP{p}{2M+1} = {\rm const} \cdot \sqrt{z} + \text{lower order terms}$.
\item It is straightforward to prove that there are no other cases where $ \alpha^{(p)}_{L}=\alpha^{(p)}_{L'}$ in the finite sheet structure. Setting $ \alpha^{(p)}_{2K+1}=\alpha^{(p)}_{2K'+1}$ and using \eref{eqA11} gives
\beq p-1 =\frac {\sin(\half( K +K'+3)\theta)}{      \sin(\half( K +K'+1)\theta )} .  \eeq
Using \eref{eqB5},\eref{eqB6} or \eref{eqB8} as appropriate and considering each in turn, it can be shown there are no $ K\ne K'$ and both in the physical range.
\item For a given $q$ the number of sheets for $\GP{p}{}$ is the same as for $\GO{}$ if $p$ is an allowed value as shown in for Case 1 and Case 2. It follows from this and (i) and (ii) immediately above that the degree of the discriminant is the same as well.

\end{enumerate}

\nocite{*}

\section*{References}

\bibliography{ref.bib}

\end{document}